\RequirePackage{amsmath}
\documentclass[runningheads,a4paper]{llncs}

\usepackage{amssymb}
\usepackage{graphicx}
\usepackage{color}

\usepackage{url}
\urldef{\mailsa}\path|{Jeannette.Janssen,Joshua.Feldman}@dal.ca|    
\newcommand{\keywords}[1]{\par\addvspace\baselineskip
\noindent\keywordname\enspace\ignorespaces#1}

\begin{document}

\mainmatter  

\title{High Degree Vertices and Spread of Infections in Spatially Modelled Social Networks}

\titlerunning{High Degree Vertices and Spread of Infections}

\author{Joshua Feldman%
\and Jeannette Janssen}
%

\institute{Dept.~of Mathematics and Statistics, Dalhousie Un.\\
Halifax, NS, Canada\\
\mailsa\\}

\maketitle

\begin{abstract}
We examine how the behaviour of high degree vertices in a network affects whether an infection spreads through communities or jumps between them.
We study two stochastic susceptible-infected-recovered (SIR) processes and represent our network with a spatial preferential attachment (SPA) network. In one of the two epidemic scenarios we adjust the contagiousness of high degree vertices so that they are less contagious. We show that, for this scenario, the infection travels through communities rather than jumps between them. We conjecture that this is not the case in the other scenario, when contagion is independent of the degree of the originating vertex. Our theoretical results and conjecture are supported by simulations.

\end{abstract}

\section{Introduction}
While community structure plays an important role in the spread of infections \cite{Salathe2010}, there are few analytic results on the topic and it is unclear precisely how clustering interacts with other network properties. Part of the difficulties in this area stem from the notion that communities consist of disjoint groups or small cycles. Recently, however, many have taken an approach to studying community structure by embedding vertices in a metric space \cite{Aiello,Flaxman2006,Hoff2002,Jacob2015,Newman2015}. One can interpret the metric space as a feature space in which nearby vertices have more affinity than the vertices at a distance and, accordingly, close vertices have a higher probability of being connected. Since a community is a group of individuals who share some similarities, we represent communities as geometric clusters. The use of spatial networks allows for a more nuanced notion of community where groups can overlap and boundaries are fuzzy. Not only is this approach more realistic, but it is also easier to analyze. We will exploit the mathematical tractability of a spatial model to study the interaction between community structure and the spread of infections.
\par
Specifically, our focus will be how the behaviour of high degree nodes changes whether an infection spreads through communities or jumps between them. This question is important because understanding who spreads diseases between communities can help guide interventions. For example, as \cite{Salathe2010} show, vaccinating hosts who bridge communities can be more effective than vaccinating highly connected individuals. If high degree vertices connect communities, then these two strategies amount to the same thing. 
\par
We model our infection with a stochastic susceptible-infected-recovered (SIR) process in discrete time. Susceptible vertices can be infected in the future, infected vertices are currently sick, and recovered vertices have gained immunity from a previous infection. To control the behaviour of high degree vertices, we recognize that infections spread through contacts (i.e. sexual contact, airborn contact, etc.), but a network edge only refers to the $\textit{potential}$ for contacts. \cite{Nordvik2006} demonstrate that the number of contacts made with a neighbour has a significant effect on epidemic dynamics. We consider two scenarios in which the ``popular" vertices behave differently: (A) high degree and low degree vertices make the same number of total contacts per time step, so highly connected vertices make fewer contacts with any individual neighbour and (B) the time vertices spend with all their neighbours per time step is proportionate to their degree, so each vertex  gets an equal amount of time with any individual neighbours. In scenario A, high degree nodes have many weak connections and, in scenario B, high degree and low degree vertices have connections of equal strength so high degree vertices should pass on the disease to more neighbours. We note that contacts are not reciprocal---two vertices can make a different number of contacts with one another. To model the relationship between contacts per neighbour per time step and the probability of infecting a susceptible neighbour in a time step, we use an STI model developed by Garnett and Anderson in \cite{Garnett1996}.
\par 
To model our network, we use the Spatial Preferential Attachment (SPA) Model proposed in \cite{Aiello}. These networks are sparse power-law graphs with positive clustering coefficients \cite{Aiello}. It has been shown that the SPA model fits real-life social networks well \cite{Janssen2012}. Since we are working with a stochastic process on a random network, we modify the SPA model to remove some randomness. We show that, compared to the SPA model, vertices in our modified networks have the same expected degree and the overall degree distribution has the same power-law coefficient.
\par
Using techniques developed by \cite{Janssen2015}, we will show that when popular high degree nodes are less infectious (scenario A), the infection will travel slowly through the metric space and respect the community structure. We also conjecture that in scenario B, the infection will take long jumps between communities. To support our result and conjecture, we run simulations on the modified SPA model. While there are numerical studies exploring the relationship between the behaviour of highly connected individuals and the spread of disease with respect to community structure \cite{Salathe2010}, we believe these are the first analytic results on the topic.
\par
Our research presents new strategies for understanding communities in networks. Networks generated by the modified SPA model display many properties of real-world systems, and are more tractable than those generated by the original SPA model. More generally, by representing communities with a continuous feature space rather than with disjoint subsets, we can easily leverage geometric properties to prove otherwise difficult results. The work presented here develops techniques for understanding this geometric conception of community structure in networks.

\section{Definitions}

Here we present the definitions of the SPA model, the modified SPA model, and our model of the infection.

\subsection{The SPA Model}

The SPA model was first proposed in \cite{Aiello}.  It is a spatial digraph model where vertices are embedded in a metric space. The metric space represents the feature space, which reflects the attributes of the vertices that determine their linking patterns. The model indirectly incorporates the principle of {\sl preferential attachment}, first proposed by Barabasi and Albert (BA) \cite{Barabasi1999}, through the notion of {\sl spheres of influence} around every vertex.
\par
Vertices with a larger in-degree have a sphere of influence with greater volume, but as time progresses the spheres of influence of all nodes decrease. In the BA model, the preferential attachment came from a probability of a newcomer connecting to the old vertices. In the SPA model, we use the sphere of influence to incorporate the preferential attachment process. If a newcomer falls within an older vertex's sphere of influence, we connect them.
\par
Specifically, vertices are embedded in a hypercube $C$ of dimension $d$ with unit volume. We endow the  hypercube with the torus metric derived from any of the $L_p$ norms. The torus metric is used to avoid edge effects. If $x$ and $y$ are two vertices in $C$, the distance between them is given by:

\begin{equation*}
	d(x, y) = \min \{ \|x - y + u\|_p: u \in \{-1, 0, 1\}^m\}
\end{equation*}

The SPA model has parameters $A_1 \in (0,1)$ and $A_2 \in [0,\infty)$. (The original SPA model also has a parameter $p$ representing the conditional probability that a vertex which falls into the sphere of influence of vertex, actually links to that vertex. We here assume this to be 1.)
\par
The model consists of a stochastic sequence of $n$ graphs  $\{G_t=(V_t,E_t)\}_{0\leq t\leq n}$ with $V_t \subset C$. The index $t$ is interpreted as the $t$-th time step. At each time $t$, the sphere of influence $S(v,t)$ of a vertex $v\in V_t$ is the ball centered at $v$ with volume
\begin{equation}
\label{def:S(v,t)}
	|S(v,t)| =  \min\left\{\frac {A_1\deg^{-}(v,t) + A_2}{t},1\right\}
\end{equation}

Let $G_0$ be the null graph. Given $G_{t-1}$, we define $G_t = (V_t,E_t)$ as follows. $V_t=V_{t-1} \cup \{v_t\}$ where $v_t$ is placed uniformly at random in $C$. The edge set $E_t=E_{t-1} \cup \{(v_t,u)\mid v_t\in S(u,t)\}$.
\par
We now review the relevant properties of the SPA model. As shown in \cite{Aiello}, the network has a power law degree distribution, with an exponent of $1+\frac{1}{A_1}$. The geometric nature of the network implies that there is a high amount of local clustering \cite{Janssen2013}. In \cite{Cooper2014}, logarithmic bounds on the directed diameter were given. In \cite{Janssen2015} it was shown that the effective undirected diameter is also logarithmically bounded.

\subsection{The Modified SPA Model}

While the SPA model generates spatial graphs that fit empirical networks well, we must modify the model to make it mathematically tractable for our purposes. Since the behaviour of a vertex during its early life significantly determines its late time behaviour, working with the preferential attachment model can be difficult. Instead, we modify the sphere of influence to depend upon the deterministic expected in-degree, instead of the stochastic actual in-degree. First we present a theorem from \cite{Janssen2013} on the expected in-degree of a vertex in the SPA model.

\begin{theorem}
Let $\omega=\omega(t)$ be any function tending to infinity together with $t$. The expected in-degree at time $t$ of a vertex $v_i$ born at time $i \geq \omega$ is given by
\[
\mathbb{E}(\deg^{-}(v_i,t))=(1+o(1))\frac{A_2}{A_1}\left(\frac{t}{i}\right)^{A_1}-\frac{A_2}{A_1}
\]
\end{theorem}

A vertex's birth time $i$ and the size of the overall network $t$ determines its expected in-degree. Furthermore, if the expected degree is larger than $\log^2 n$ then the real in-degree is concentrated around its expected in-degree \cite{Janssen2013}. Hence, the time a vertex is born can be used as a proxy for its degree, with old nodes being more highly connected than young nodes. 
\par
We modify the SPA model by redefining the sphere of influence to depend on the vertex's expected in-degree, instead of its actual in-degree. This substitution makes the size of the sphere of influence a nonrandom object. Precisely, the $\emph{modified spatial preferential attachment}$ model is defined as the SPA model, with the one difference being that the size of the sphere of influence of vertex $v_i$ at time is changed to: 

\begin{equation}
	|S(v_i,t)| =  \min\left\{\frac{A_2}{t^{1-A_1}i^{A_1}},1\right\}	
\end{equation}

We derive equation (2) by replacing the actual in-degree in formula (1) for the original sphere of influence with the expected in-degree and simplifying. We now state and prove a theorem which shows that the modified SPA model generates networks with an expected in-degree equivalent to the original model.

\begin{theorem}
Let $G_n$ be a graph generated by the modified SPA model with $n$ vertices. The expected in-degree of a vertex $v_i$ born at time $i$ is given by
\[
\mathbb{E}(\deg^{-}(v_i,n))=\frac{A_2}{A_1}\left(\left(\frac{n}{i}\right)^{A_1}-1\right)-\epsilon
\]
where $|\epsilon|<\frac{A_2}{A_1}$.
\end{theorem}

\begin{proof}
Let $v_j$ be a vertex born at time $j>i$. Let $X_j$ be a random variable that equals $1$ if there is an edge from $v_j$ to $v_i$ and equals $0$ otherwise. By the definition of the modified SPA model, we place an edge from $v_j$ to $v_i$ if and only if $v_j$ falls within the sphere of influence of $v_i$. Since $v_j$ is placed in the hypercube uniformly at random, $X_j=1$ with probability equal to the volume of $v_i$'s sphere of influence at time $j$. 
\par
By the linearity of expectation,
\begin{align*}
	\mathbb{E}(\deg^{-}(v_i,n)) &=  \mathbb{E}\left(\sum_{k=i+1}^{n} X_k\right) \\
			      &=\sum_{k=i+1}^{n}  \mathbb{E}\left(X_k\right) \\
			      &=\sum_{k=i+1}^{n}  |S(v_i,k)| \\
\end{align*}
We approximate this sum with an integral.
\begin{align*}
	\sum_{k=i+1}^{n}  |S(v_i,k)| &= \int_{i}^{n} \frac{A_2}{x^{1-A_1}i^{A_1}} dx -\epsilon \\
						&= \frac{A_2}{A_1}\left(\left(\frac{n}{i}\right)^{A_1}-1\right) -\epsilon
\end{align*}

To bound the error, we first recognize that,
\[
	\int_{i+1}^{n+1} \frac{A_2}{x^{1-A_1}i^{A_1}} dx \leq \sum_{k=i+1}^{n}  |S(v_i,k)| \leq \int_{i}^{n} \frac{A_2}{x^{1-A_1}i^{A_1}} dx
\]

Hence,
\begin{align*}
	|\epsilon| &< \int_{i}^{n} \frac{A_2}{x^{1-A_1}i^{A_1}} dx - \int_{i+1}^{n+1} \frac{A_2}{x^{1-A_1}i^{A_1}} dx \\
			&= \frac{A_2}{i^{A_1}}\frac{((i+1)^{A_1}-i^{A_1})-((n+1)^{A_1}-n^{A_1})}{A_1} \\
			&< \frac{A_2}{i^{A_1}}\frac{(i+1)^{A_1}-i^{A_1}}{A_1} \\
			&< \frac{A_2}{A_1}\left(2^{A_1}-1\right) \\
			&< \frac{A_2}{A_1}
\end{align*}
\qed
\end{proof}

In addition to having equivalent expected in-degrees, we also derive that both models lead to the same (power law) cumulative in-degree distribution. The cumulative in-degree distribution $c_k$ is defined as the number of vertices with in-degree $j \leq k$ divided by the total number of vertices. As stated above, from \cite{Aiello} we know that a.a.s. the cumulative in-degree distribution of networks generated by the SPA model follows a power law with exponent $\frac{1}{A_1}$. Theorem 3 states that the same is true of networks generated by the modified SPA model. An event occurs \textit{with extreme probability (w.e.p.)} if it occurs with probability at least $1-e^{-\Theta(\log^2n)}$ as $n \rightarrow \infty$.

\begin{theorem}
Let $G_n$ be a graph generated by the modified SPA model with $n$ vertices. The cumulative in-degree distribution $c_k$ of $G_n$ is w.e.p. a power law with exponent $\frac{1}{A_1}$ for $k > k' = \log^2(n)$.
\end{theorem}

\begin{proof}
The in-degree of a vertex $v_i$ born at time $i$ is the sum of $n-i$ independent Bernoulli variables $X_j$ with success probability equal to the volume of the sphere of influence $|S(v_i,j)|$. We let $f(i)$ equal the expected in-degree of a vertex born at time $i$ in a network of size $n$. By the generalized Chernoff bound \cite{Lu2006}, we know that w.e.p. $f(i) - \epsilon \leq \deg^{-}(v_i,n) \leq f(i) + \epsilon$ where $\epsilon=\sqrt{f(i)}\log n$.
\par
Using this bound, we determine how many vertices have in-degree greater than $k$. If a vertex is born before $i^-=f^{-1}(k+\epsilon)$, then w.e.p. it has a degree greater than $k$. Likewise, if a vertex is born after $i^+=f^{-1}(k-\epsilon)$, then w.e.p. it has a degree less than $k$. Hence, the number of vertices with degree greater than $k$ is between $i^-$ and $i^+$. By inverting the formula for expected degree and examining its asymptotic growth, we find
\begin{equation*}
i^-=f^{-1}(k+\epsilon)=(1+o(1))f^{-1}(k)
\qquad
i^+=f^{-1}(k-\epsilon)=(1+o(1))f^{-1}(k)
\end{equation*}
Hence, the number of vertices with degree greater than $k$ is w.e.p. $(1+o(1))f^{-1}(k)$. This implies that $c_k=(1+o(1))(kA_1/A_2-1)^{-1/A_1}$.
\qed
\end{proof}

\subsection{Infectious Processes}

Now that we have a workable model of real-world networks, we define a SIR disease model in discrete time. To begin the infectious process, we pick a node at random to be the origin node. At time $t=0$, we infect the origin node and denote all other nodes as susceptible. In each time step, the infected nodes infect each neighbour with probability $\beta$. Though the modified SPA model generates directed graphs, we ignore the orientation of the edges. If a susceptible vertex has multiple infected neighbours, they each independently attempt to infect the susceptible vertex. At the end of each time step, all infected nodes recover. This is a simplification of the typical SIR model because usually the recovery time is modelled as a stochastic variable. Here we simplify and assume each vertex to recover in exactly one time step. We run the process until no vertices are infected. If we run the infection process for $t$ time steps on a network with $n$ vertices, the infected and recovered vertices together with the edges taken by the infection (oriented in the direction the infection travelled) form an an acyclic directed subgraph of the network. We denote this subgraph, $I_n^t$, the \textit{infection graph at time $t$}.

\par
Suppose that in a given time step vertex $v$ is infected, vertex $u$ is susceptible, and they are neighbours. The probability $\beta$ of $v$ infecting $u$ in the time step depends on $\kappa(v)$, the average number of contacts $v$ makes with $u$ per time step, and the probability of transmitting the infection per contact $\tau$. If $v$ makes more contacts on average with $u$ or if the probability of infection per contact is higher, the disease will have a greater chance of spreading. Following \cite{Garnett1996}, we set 
\[
\beta=1-e^{-\tau\kappa}
\]

\par
To study how the behaviour of high degree nodes changes how a disease spreads through the network, we consider two different scenarios: scenario A and scenario B with their own infection probabilities $\beta_A$ and $\beta_B$, respectively. In scenario A, we define the average number of contacts a vertex $v$ makes with a neighbour per time step as
\[
\kappa_A(v)=\frac{T}{\mathbb{E}(\deg^{-}(v))}
\]
where $T$ is the average number of contacts $v$ makes with \textit{all} its neighbours in the time step. We use $\mathbb{E}(\deg^{-}(v))$ as a rough approximation of the degree of $v$. Hence, in scenario A, $\beta_A(v)=1-e^{-\tau\kappa_A(v)}$. Since $T$ is constant for all vertices, high and low degree alike, high degree vertices make fewer contacts with any $\textit{single}$ neighbour because their contacts are dispersed over a greater number vertices. While high degree nodes have many neighbours, these connections may be weaker than a node with few neighbours. 
\par
In scenario B, however, we no longer keep the average number of contacts a vertex makes with all its neighbours in a time step constant. Instead, we let $T(v)$ depend on $T$ and the expected in-degree of $v$. We define 
\[
T(v)=T\frac{\mathbb{E}(\deg^{-}(v))}{\langle\deg^-\rangle}
\]
where $\langle\deg^-\rangle$ is the average degree in the graph. From \cite{Aiello}, we know a.a.s. that $\langle\deg^-\rangle=(1+o(1))\frac{A_2}{1-A_1}$ in the SPA model, which is asymptotically constant. Since the total number of edges in the network is the sum of Bernoulli variables, by the linearity of expectation, it is a simple calculation to show that $\langle\deg^-\rangle$ is equivalent in the modified and original SPA models. As in scenario A, we use $\mathbb{E}(\deg^{-}(v))$ to approximate the degree of $v$. From the average number of contacts a vertex makes with \textit{all} its neighbours in a time step, we can define the average number of contacts a vertex makes with a \textit{specific} neighbour in a time step as 
\[
\kappa_B(v) =\frac{T(v)}{\mathbb{E}(\deg(v))}=\frac{T}{\langle\deg^-\rangle}
\]
Hence, in scenario B, $\beta_B=1-e^{-\tau\kappa_B}$, which is constant. Any infected vertex $v$ has an equal chance of infecting a neighbour, regardless of the degree of $v$. Hence, in scenario B, we expect that high degree vertices will pass the infection on to more individuals than low degree vertices.

\section{Spatial spread of infections}

Our main result states that when highly connected vertices are less infectious, the infection will not make large jumps through the metric space. Since our metric space represents a feature space, this means that the infection spreads through communities rather than jumping between them. To prove this result, we treat the infection as percolating through the network. We first show that a.a.s. all vertices in the network will only infect neighbours within a certain distance. From this result, we conclude that any particular infection will a.a.s. be bounded by a ball of a relatively small radius after a given number of time steps.
 
\begin{theorem}
\label{thm:main}
Let $G_n$ be a graph with $n$ vertices generated by the modified SPA model. Let $\lambda=n^{-\phi}$ be such that $\phi < \frac{A_1(1-A_1)}{(A_1+2)d}$. For scenario A, a.a.s. all nodes in the infection graph at time $t$ will be within  $t\lambda$ of the origin node $u$. 
\end{theorem}

\subsection{Proof of Theorem \ref{thm:main}}

Before we present the proof, we first adopt some conventions regarding the infection process. Instead of considering the infection spreading through a network in time, we $\textit{a priori}$ consider whether any vertex would infect a neighbour given that they are connected. If we ``occupy" each pair of vertices with probability $\beta_A$, and restrict our occupied pairs to edges present in our network, we get a subgraph consisting of where the infection could \textit{possibly} travel, which we call the potential infection graph. The infection graph, describing where the infection \textit{actually} spread, will be a subset of the potential infection graph.
\par
Formally, let $G_n=(V_n,E_n)$ be a network of order $n$ generated by the modified SPA model, where we replace each directed edge by two edges in opposite directions. We consider ordered pairs of vertices $(v_i,v_j)$ and $(v_j,v_i)$ because our infection model ignores the orientation of the edges in the original network generated by the modified SPA model. In other words, even though all edges in the modified SPA model are directed from younger to older vertices, we allow the old to infect the young. Let $u \in V_n$ be the node where the infection originates. With each ordered pair of vertices $(v_i,v_j)$ we associate a Bernoulli random variable $I_{v_i,v_j}$ defined as
\begin{equation}
I_{v_i,v_j}=
\begin{cases} 
      1  
      & \text{with probability } \beta_A(v_i)  \\
      0 
      & \text{otherwise}\\
\end{cases}
\end{equation}
We define the \textit{potential infection graph} on $G_n$ as the graph $I_n=(V_I,E_I)$ where $V_I=V_n$ and $E_I=\{ (v_i,v_j) | \{v_i,v_j\} \in E_n  \text{ and } I_{v_i,v_j}=1\}$.  We can recover the infection graph at time $t$, $I_n^t$, by taking the subgraph induced by the $t^\text{th}$ out-neighbourhood of $u$ in $I_n$.
\par
The proof of our main result is based on an analysis of the edges in the potential infection digraph. Define the \emph{length of an edge} as the distance between its two end points. We first establish a lemma stating that there is an asymptotic bound on the length of edges in the potential infection graph.

\begin{lemma}
\label{lem:NoLongEdges}
Let $G_n$ be a graph with $n$ vertices generated by the modified SPA model and $I_n$ be a potential infection graph on $G_n$ in scenario A. Let $\lambda=n^{-\phi}$ such that $\phi < \frac{A_1(1-A_1)}{(A_1+2)d}$. Then a.a.s. $I_n$ does not contain any edges of length greater than $\lambda$. 
\end{lemma}

\begin{proof}
Let $L$ represent the event of there being an edge in $I_n$ where the distance between its endpoints is greater than $\lambda$. We will call such edges ``long" and all other edges ``short". Given two (not necessarily connected) nodes in $V_n$, $v_i$ and $v_j$, let $L_{v_i,v_j}$ represent the event of there being a long edge from $v_i$ to $v_j$ in $I_n$. Thus, $L_{v_i,v_j}$ occurs if $v_i$ and $v_j$ (the vertices born at time $i$ and $j$, respectively) have distance at least $\lambda$, there is an edge between $v_i$ and $v_j$, and the infection can travel from $v_i$ to $v_j$.  Since $L=\bigcup_{i=0}^{n-1}\bigcup_{j=i+1}^{n} \left(L_{v_i,v_j}\cup L_{v_j,v_i}\right)$, by taking the union bound, we know
\[
	\mathbb{P}(L)\leq \sum\limits_{i=0}^{n-1}\sum\limits_{j=i+1}^{n} \mathbb{P}(L_{v_i,v_j})+\mathbb{P}(L_{v_j,v_i})
\]
Our proof will show that this double sum goes to $0$ as $n$ approaches infinity.

\par

We first need an expression for $\mathbb{P}(L_{v_i,v_j})+\mathbb{P}(L_{v_j,v_i})$. Since $i<j$, by the definition of the potential infection graph, $L_{v_i,v_j}$ occurs if and only if three other events also occur: $d(v_i,v_j) > \lambda$, $v_j \in S(v_i,j)$, and $I_{v_i,v_j}=1$. In other words, for there to be a long edge between $v_i$ and $v_j$, $I_{v_i,v_j}$ must equal $1$ and $v_j$ must be far enough away from $v_i$ to be considered long, but close enough to be in the sphere of influence of $v_i$ at time $j$. 

\par

Since $v_j$ is placed uniformly at random in the hypercube, the distance $d(v_j,v_i)$ and the event $I_{v_i,v_j}=1$ are independent. Therefore, for any specific values for $j$ and $i$. $i<j$,  we can write

\begin{align*}
	\mathbb{P}(L_{v_i,v_j}) &= \mathbb{P}(d(v_i,v_j) > \lambda, v_j \in S(v_i,j), I_{v_i,v_j}=1) \\
			      &= \mathbb{P}(d(v_i,v_j) > \lambda, v_j \in S(v_i,j))\mathbb{P}(I_{v_i,v_j}=1)
\end{align*}

\par
For the edge oriented in the opposite direction, we can make a similar argument. Hence, we can write

\begin{align*}
	\mathbb{P}(L_{v_j,v_i}) &= \mathbb{P}(d(v_i,v_j) > \lambda, v_j \in S(v_i,j), I_{v_j,v_i}=1) \\
			      &= \mathbb{P}(d(v_i,v_j) > \lambda, v_j \in S(v_i,j))\mathbb{P}(I_{v_j,v_i}=1)
\end{align*}

Combining our expressions for $\mathbb{P}(L_{v_i,v_j})$ and $\mathbb{P}(L_{v_j,v_i})$ we find that

\[
	\mathbb{P}(L_{v_i,v_j})+\mathbb{P}(L_{v_j,v_i}) = \mathbb{P}(d(v_i,v_j) > \lambda, v_j \in S(v_i,j))(\mathbb{P}(I_{v_i,v_j}=1)+\mathbb{P}(I_{v_j,v_i}=1))
\]

 We know $\mathbb{P}(I_{v_i,v_j}=1)=\beta_A(v_i)$ and $\mathbb{P}(I_{v_j,v_i}=1)=\beta_A(v_j)$ from equation (3), but we need expressions for $\mathbb{P}(d(v_i,v_j) > \lambda, v_j \in S(v_i,j))$ and $\mathbb{P}(I_{v_i,v_j}=1)$. 

We use a geometric argument to find $\mathbb{P}(d(v_j,v_i) > \lambda, v_j \in S(v_i,j))$, which is the probability of there being a long edge between $v_i$ and $v_j$ in the original network. There are three cases. In the first case, the sphere of influence of $v_i$ has radius smaller than $\lambda$ at its time of birth. Since spheres of influence only shrink, in this case there will never be a time when $v_j$ can both fall within $v_i\text{'s}$ sphere of influence and be greater than $\lambda$ away from $v_i$. This case occurs when $i$ exceeds a critical value $m$, which is the first time when vertices are born with a sphere of influence that  has radius smaller than $\lambda$.

In the second case, $i$ is smaller than the critical value $m$, but $j$ is larger than the second critical value $m_i$. This critical value is reached when the radius of $v_i\text{'s}$ sphere of influence equals $\lambda$. Again, since spheres of influence only shrink, vertices born at times after $m_i$ cannot have $d(v_j,v_i) > \lambda$ and  $v_j \in S(v_i,j))$. In these first two cases, $\mathbb{P}(d(v_j,v_i) > \lambda, v_j \in S(v_i,j))=0$. 

A ball of radius $\lambda$ has volume $\lambda^dc_p$ where $c_p$ depends on our $L_p$ norm. 
Using this, we find that
\begin{equation*}
m=\frac{A_2}{\lambda^dc_p}
\qquad
m_i=\left(\frac{A_2}{i\lambda^dc_p}\right)^{\frac{1}{1-A_1}}
\end{equation*}

In the last case, $i<m$ and $j<m_i$, which means $d(v_j,v_i) > \lambda$ and  $v_j \in S(v_i,j))$ is possible. Since $v_j$ is placed in the hypercube uniformly at random and the hypercube has unit volume, $\mathbb{P}(d(v_j,v_i) > \lambda, v_j \in S(v_i,j))$ is equal to the volume of the spherical shell between the sphere of influence and the ball centred at $v_i$ with radius $\lambda$ which we denote $B(v_i,\lambda)$. Hence, in this case, $\mathbb{P}(d(v_j,v_i) > \lambda, v_j \in S(v_i,j))=|S(v_i,j)|-|B(v_i,\lambda)|$.

\par

Combining the results from the previous paragraphs,

\[
\mathbb{P}(L) \leq \sum\limits_{i=0}^{m}\sum\limits_{j=i+1}^{m_i} (|S(v_i,j)|-|B(v_i,\lambda)|)(\beta_A(v_i) +\beta_A(v_j)) 
\]

Since the oldest vertex will always have the largest sphere of influence, we know that $m_i < m_1$ for all $i \in [1,m]$ and that $m<m_1$. Also, $A_2 = |S(v_1,1)| > |S(v_i,j)-B(v_i,\lambda)|$ for all $i,j \in [1,m_1]$. Finally, since $v_{m_1}$ has the lowest expected degree of all vertices born at or before time $m_1$, $\beta_A(v_{m_1}) > \beta_A(v_i),\beta_A(v_j)$ for all $i,j \in [1,m_1]$. Hence, we can write

\begin{align}
	\mathbb{P}(L) &\leq  \sum\limits_{i=1}^{m_1}\sum\limits_{j=i+1}^{m_1} 2A_2\beta_A(v_{m_1}) \\
			      &\leq 2\left(1-\exp\left(-\frac{\tau T}{\frac{A_2}{A_1}\left(\left(\frac{n}{{m_1}}\right)^{A_1}-1\right)}\right)\right)A_2m_1^2 
\end{align}
From the formula for $m_1$, we see that $m_1 \sim n^{\frac{\phi d}{1-A_1}}$. Setting $\phi = \frac{A_1(1-A_1)}{(A_1+2)d}$ and $\gamma=\tau T$, we see that

\begin{align}
	\mathbb{P}(L) &\leq 2\left(1-\exp\left(-\frac{\gamma}{\left(\left(\frac{n}{{m_1}}\right)^{A_1}-1\right)}\right)\right)A_2m_1^2 \\
			      &\sim 2\left(1-\exp\left(-\gamma n^{A_1\left(\frac{\phi d}{1-A_1}-1\right)}\right)\right)A_2n^{\frac{2\phi d}{1-A_1}} \\
			      &\sim 2\left(1-\left(1-\gamma n^{A_1\left(\frac{\phi d}{1-A_1}-1\right)} + O\left(n^{2A_1\left(\frac{\phi d}{1-A_1}-1\right)}\right)\right)\right)A_2n^{\frac{2\phi d}{1-A_1}} \\
			      &= o(1)
\end{align}
\qed
\end{proof}

Using this lemma, we can now prove our main result, Theorem 4.

\begin{proof}
Let B represent the bad event where there is a node $v$ in the infection digraph at time $t$, $I_n^t$, where $d(v,u)>t\lambda$. If B occurs, then there is a path from $v$ to $u$ with at most $t$ edges because $I_n^t$ is the $t^\text{th}$ neighbourhood of $u$ in $I_n$. By the triangle inequality, at least one of the edges in the path from $v$ to $u$ has a length greater than $\lambda$ and, more generally, there is an edge in the potential infection graph with a length greater than $\lambda$. Let $L$ represent the event of there being an edge in $I_n$ where the distance between its endpoints is greater than $\lambda$. Since $B \subset L$, $\mathbb{P}(B) \leq \mathbb{P}(L)$, but by the previous lemma, a.a.s. $\mathbb{P}(L) = 0$.
\qed
\end{proof}

\subsection{Conjecture for Scenario B}
We conjecture that in scenario B, the negation of Lemma 1 holds. We know that the modified SPA model a.a.s. has edges greater than length $\lambda'$ where $\lambda'=\mu n^{-\theta}$ with $\theta>1-\frac{A_1}{4A_1+2}$ and $\mu$ constant. We conjecture that the potential infection graph will have long edges as well.
\newtheorem{Conjecture}{Conjecture}
\begin{Conjecture}
Let $G_n$ be a graph with $n$ vertices generated by the modified SPA model and $I_n$ be a potential infection graph on $G_n$ in scenario B. There exists $\phi>0$ such that if we let $\lambda=n^{-\phi}$, a.a.s. $I_n$ contains an edge of length greater than $\lambda$. 
\end{Conjecture}

\section{Simulations}

Using simulations, we test our theoretic result that the infection in scenario A will not make long jumps. We also use simulations to provide evidence for our conjecture that, in scenario B, the infection can make long jumps if we pick the origin vertex correctly. Recall that in both infection scenarios, we can vary how easily the infection spreads by altering $T$, the total number of contacts made per time step, and $\tau$, the probability of transmission per contact. Also recall that we denote $\gamma = \tau T$. For the 2 infection scenarios, we consider 3 levels of contagiousness: $\gamma=1$, $\gamma=10$ and $\gamma=100$. 

We generated 10 networks with the modified SPA network in $\mathbb{R}^1$ with $A_1=0.5$ and $A_2=1$. Our results are highly asymptotic and the bound is lowest in low dimensions so, due to computational constraints, we choose to simulate in $\mathbb{R}^1$. The 10 networks are of increasing size, beginning at $n=1000$ and increasing by increments of $1000$ to a maximum of $n=10000$. 

For each network, we run each of the 6 infection processes 50 times. We chose to begin the infections at the oldest vertex because it has the highest likelihood of having neighbours far away in the metric space. On one hand, we want to give the infection in scenario B the opportunity to make long jumps and, on the other hand, we do not want to mistakenly conclude that the infection in scenario A makes short jumps only because it was never exposed to long edges. While our main result states that given a number of time steps, a.a.s. the infection remains within a certain region, this result depends on both the size of the network and the current time step in the infection process. We thought it would be more clear to compare Lemma 1 to our simulations, which states that the edge length taken by the infection in scenario A is bounded. To compare scenario A and B, we likewise observe the maximum edge length the infection in scenario B takes.

\begin{figure}
\vspace{0cm}
  \includegraphics[width=\linewidth]{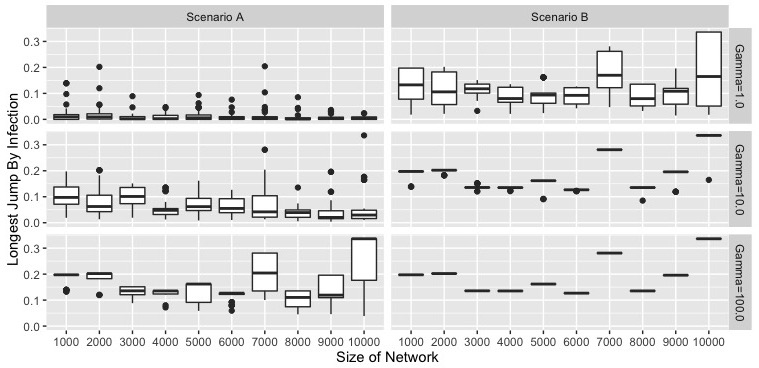}
  \caption{Longest jump made by the infection vs. network size in scenario A and B stratified by 3 levels of contagion.}
  \label{fig:SimulationGraph}
\end{figure}

The results of our simulations are shown in Figure 1. We make two conclusions from our simulations. First, we see that in scenario A, when high degree nodes are less contagious, the infection takes shorter jumps than in scenario B. As indicated by our asymptotic result, the difference becomes more pronounced in larger networks. One might notice that in scenario B the infection does not always make long jumps, which seems to contradict our conclusion. These outliers can be explained by recognizing that the longest edge in the entire network is a non-clustered random variable. Vertices receive long edges in a brief period during the early steps of the model and, consequently, whether there are long edges in the network at all is highly variable. It is not that the infection avoids the long edges, but rather, that the infection has no long edges to take in the first place. Second, we see that this difference between scenario A and scenario B becomes less pronounced when we make the infection more contagious. Again, this matches our theoretic result, since our bound on the probability of long infection increases when $\gamma$ is larger. We expect that if we could generated large enough networks, eventually we would see the difference between scenario A and B reemerge, even at high levels of contagion.

\begin{figure}
\vspace{0cm}
  \includegraphics[width=\linewidth]{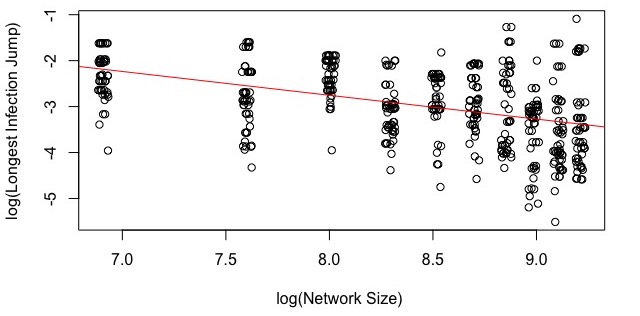}
  \caption{A log-log plot of the longest jump made by the infection vs. network size in scenario A with $\gamma=10$. The equation for the regression line is log(y)=-0.51log(x)+1.40 ($R^2=0.19$).}
  \label{fig:LogLogPlot}
\end{figure}

To compare our analytic bound $\phi$ to our simulations, we perform a linear regression on a log-log plot of longest jump vs. network size for infections in Scenario A with $\gamma=10$ (see Figure 2). The simulations are the same as those represented in Figure 1. We have added a small amount of noise to the x-values in order to make the distribution of data more clear. If our result is true, then we should expect that the slope of the regression line should be less than $-\frac{A_1(1-A_1)}{(A_1+2)d}=-0.1$. We found the slope to be $-0.51$, which provides support for our lemma. Of course, our data is highly variable and this plot only gestures towards the fact that our bound is valid. We expect that for larger networks, this variation would decrease. 

To illustrate that the infection spreads slowly through the feature space in scenario A, we simulate one run of the process on a graph generated by the modified SPA model in $\mathbb{R}^2$ with $A_1=0.5$, $A_2=1$, $\gamma=10$, and $n=1000$. We present the simulation in Figure 3. The blue vertices were infected earlier in the process and red vertices later. Since nearby vertices have similar colours (recall that we are using the torus metric), this simulation provides additional qualitative evidence for our finding that the infection does not make long jumps in scenario A.
\begin{figure}
\vspace{0cm}
  \includegraphics[width=\linewidth]{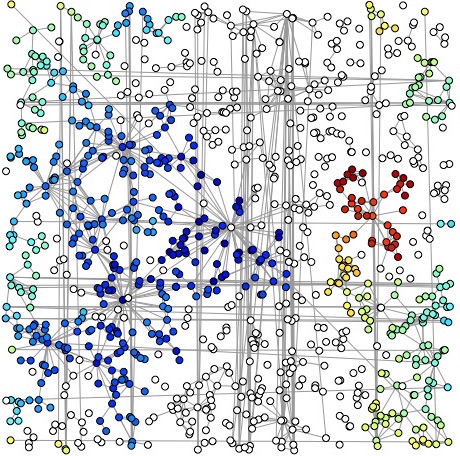}
  \caption{A scenario A infection with $\gamma=10$ on a modified SPA network with $A_1=0.5$, $A_2=1$, and $n=1000$. The gradient from blue to red represents earlier to later infections.}
  \label{fig:SimulationGraph}
  \setlength{\belowcaptionskip}{-10pt}
\end{figure}
\section{Conclusion}
When modelling contagious processes, it is important to take contacts made per neighbour into account. With analytic and numeric results, we show that if all vertices make an equal number of contacts, the infection will spread through communities rather than jumping between them. High degree vertices are more likely to have neighbours in distant communities and when these popular individuals are less contagious, the infection is less likely to spread from one community to another. We also show with simulations that scaling the number of contacts a vertex makes by its degree results in an epidemic that spreads irrespective of the communities in the network.

In addition to our conjecture, we identify two areas of future research. First, since infections in scenario A and scenario B behave differently with respect to community structure, interventions may benefit from exploiting this distinction. In other words, if we know a contagious process will spread through communities, how can we use this fact to control the epidemic? Likewise, how should we control diseases that jump between communities? The second area of potential research is studying scenario A infections further. While we find that these types of contagious processes will be a.a.s. bounded by a ball with a growing radius, we do not discuss how this may affect the success of an infection spreading through a network. If a disease does not jump, does community structure prevent the infection from spreading beyond its original group?
\bibliographystyle{splncs03}
\bibliography{SPARef}

\begin{thebibliography}{10}
\providecommand{\url}[1]{\texttt{#1}}
\providecommand{\urlprefix}{URL }

\bibitem{Aiello}
Aiello, W., Bonato, A., Cooper, C., Janssen, J., Pralat, P.: {A Spatial Web
  Graph Model with Local Influence Regions}. Internet Mathematics  5(12),
  173--193 (2007)

\bibitem{Barabasi1999}
Barab{\'{a}}si, A.L., Albert, R.: {Emergence of Scaling in Random Networks}.
  Science  286(5439) (1999)

\bibitem{Cooper2014}
Cooper, C., Frieze, A., Pra{\l}at, P.: {Some Typical Properties of the Spatial
  Preferred Attachment Model}. Internet Mathematics  10(1-2),  116--136 (2014)

\bibitem{Flaxman2006}
Flaxman, A.D., Frieze, A.M., Vera, J.: {A Geometric Preferential Attachment
  Model of Networks}. Internet Mathematics  3(2),  187--206 (2006)

\bibitem{Garnett1996}
Garnett, G.P., Anderson, R.M.: {Sexually Transmitted Diseases and Sexual
  Behavior: Insights from Mathematical Models}. The Journal of Infectious
  Diseases  174(2),  S150--S161 (1996)

\bibitem{Hoff2002}
Hoff, P.D., Raftery, A.E., Handcock, M.S.: {Latent Space Approaches to Social
  Network Analysis}. Journal of the American Statistical Association  97(460),
  1090--1098 (2002)

\bibitem{Jacob2015}
Jacob, E., Morters, P.: {Spatial preferential attachment networks: Power laws
  and clustering coefficients}. Annals of Applied Probability  25(2),  632--662
  (2015)

\bibitem{Janssen2012}
Janssen, J., Hurshman, M., Kalyaniwalla, N.: {Model Selection for Social
  Networks Using Graphlets}. Internet Mathematics  8(4),  338--363 (2012)

\bibitem{Janssen2015}
Janssen, J., Mehrabian, A.: {Rumours Spread Slowly in a Small World Spatial
  Network}. Algorithms and Models for the Web Graph pp. 107--118 (2015)

\bibitem{Janssen2013}
Janssen, J., Pralat, P., Wilson, R.: {Geometric graph properties of the spatial
  preferred attachment model}. Advances in Applied Mathematics  50(2),
  243--267 (2013)

\bibitem{Lu2006}
Lu, L., Chung, F.: {Old and new concentration inequalities}. In: Complex Graphs
  and Networks, chap.~2, pp. 23--56. American Mathematical Society, Providence
  (2006)

\bibitem{Newman2015}
Newman, M.E.J., Peixoto, T.P.: {Generalized communities in networks}. Physical
  Review Letters  115(8) (2015)

\bibitem{Nordvik2006}
Nordvik, M.K., Liljeros, F.: {Number of Sexual Encounters Involving Intercourse
  and the Transmission of Sexually Transmitted Infections}. Sexually
  Transmitted Diseases  33(6),  342--349 (2006)

\bibitem{Salathe2010}
Salath{\'{e}}, M., Jones, J.H., May, R., Johnson, A., Auranen, K.: {Dynamics
  and Control of Diseases in Networks with Community Structure}. PLoS
  Computational Biology  6(4),  e1000736 (2010)

\end{thebibliography}
\end{document}